\documentclass[12pt]{article}

\usepackage{cleveref}

\usepackage{amsmath}
\usepackage{amsthm}
\usepackage{amssymb}
\usepackage{csquotes}
\usepackage[dvipsnames]{xcolor}
\usepackage{thmtools}
\usepackage{thm-restate}

\usepackage{cite}

\newcommand{\defn}[1]{\emph{\textbf{\boldmath #1}}}

\newcommand{\E}[1]{\mathrm{E}\left[#1\right]}

\newtheorem{lemma}{Lemma}
\newtheorem{theorem}{Theorem}
\newtheorem{corollary}{Corollary}

\newcommand{\boa}{BOA}

\newcommand{\Dinsert}[3]{\textsc{insert}$(#1,#2,#3)$}
\newcommand{\Ddelete}[2]{\textsc{delete}$(#1,#2)$}
\newcommand{\Dquery}[2]{\textsc{query}$(#1,#2)$}
\newcommand{\Dsucc}[2]{\textsc{succ}$(#1,#2)$}
\newcommand{\Dpred}[2]{\textsc{pred}$(#1,#2)$}

\newcommand{\calS}{\mathcal{S}}
\newcommand{\calU}{\mathcal{U}}
\newcommand{\calV}{\mathcal{V}}

\newcommand{\bet}{B$^{\varepsilon}$-tree}

\newcommand{\grow}{\lambda}

\def\hyperspc{\kern -0.22em}

\newcommand{\IInsist}[1]{\smallskip\noindent{\sffamily\normalsize\bfseries #1.}}

\begin{document}

\title{Optimal Hashing in External Memory\thanks{Supported by NSF CCF 1637458, NIH 1 U01 CA198952-01, a NetAPP Faculty Fellowship and a gift from Dell/EMC.}}

\author{Alex Conway\\Rutgers
  University\\\texttt{alexander.conway@rutgers.edu} 
\and
Mart\'{\i}n Farach-Colton\\Rutgers University
\\\texttt{farach@rutgers.edu}
\and
Philip Shilane\\Dell EMC \\texttt{shilane@dell.com}}

\maketitle
\begin{abstract}
	Hash tables are a ubiquitous class of dictionary data structures. However,
	standard hash table implementations do not translate well into the external
	memory model, because they do not incorporate locality for insertions.

	Iacono and P\u atra\c su established an update/query tradeoff curve for
	external hash tables: a hash table that performs insertions in
	$O(\lambda/B)$ amortized IOs requires $\Omega(\log_\lambda N)$ expected IOs
	for queries, where $N$ is the number of items that can be stored in the
	data structure, $B$ is the size of a memory transfer, $M$ is the size of
	memory, and $\lambda$ is a tuning parameter.
	
	They provide a hashing data structure that meets this curve for $\lambda$
	that is $\Omega(\log\log M + \log_M N)$.  Their data structure, which
	we call an \defn{IP hash table}, is complicated and, to the best of our
	knowledge, has not been implemented.

	In this paper, we present a new and much simpler optimal external memory
	hash table, the \defn{Bundle of Arrays Hash Table} (BOA). BOAs
	are based on size-tiered LSMs, a well-studied data structure, and are
	almost as easy to implement. The BOA is optimal for a
	narrower range of $\lambda$.  However, the simplicity of BOAs
	allows them to be readily modified to achieve the following results:

\begin{itemize}
	\item A new external memory data structure, the \defn{Bundle of Trees Hash
		Table} (BOT), that matches the performance of the IP hash table, while
		retaining some of the simplicity of the BOAs.
	\item The \defn{cache-oblivious Bundle of Trees Hash Table} (COBOT), the
		first cache-oblivious hash table. This data structure matches the
		optimality of BOTs and IP hash tables over the same range of
		$\lambda$.
\end{itemize}

\end{abstract}

\section{Introduction} 

Dictionaries are among the most heavily used data structures. A dictionary
maintains a collection of key-value pairs $\calS\subseteq \calU\times \calV$, under
operations\footnote{We do not consider dictionaries that also support the
	\Dsucc{x}{\calS} and \Dpred{x}{\calS}.  \Dsucc{x}{\calS} return $\min\{y|
	y>x \land y\in\calS\}$ and \Dpred{x}{\calS} is defined symmetrically.}
\Dinsert{x}{v}{\calS}, \Ddelete{x}{\calS}, and \Dquery{x}{\calS}, which returns
the value corresponding to $x$ when $x\in S$.  When data fits in memory, there
are many solutions to the dictionary problem.

When data is too large to fit in memory, comparison-based dictionaries can be
quite varied.  They include the
\bet~\cite{DBLP:conf/soda/BrodalF03}, the write-optimized skip
list~\cite{DBLP:conf/pods/BenderFJMMPX17}, and the cache-optimized look-ahead array
(COLA)~\cite{DBLP:conf/spaa/BenderFFFKN07,DBLP:journals/pvldb/BenderFJKKMMSSZ12,DBLP:conf/esa/BenderCDF02}.
All of these data structures have been deployed extensively, and their best
variants are optimal in the \defn{external-memory comparison model} in
that they match the bound established by Brodal and
Fagerberg~\cite{DBLP:conf/soda/BrodalF03} who showed that for any
dictionary in this model, if insertions can be performed in
$O\left(\frac{\grow\log_{\grow} N}{B}\right)$ amortized IOs, then there exists
a query that requires at least $\Omega(\log_\grow N)$ IOs.  This trade off has
since been extended in several
ways~\cite{DBLP:conf/soda/AfshaniBFFGT17}.

In this paper, we consider the dictionary problem without restriction to the
comparison model, and in particular, we consider external-memory hashing.  This
allows for a better insertion/query trade off.  Iacono and P\v{a}tra\c{s}cu
showed:
\begin{theorem}[\cite{DBLP:journals/corr/abs-1104-2799}]\label{thm:ip-lower-bound}
	If insertions into an external memory dictionary can be performed in
	$O\left(\grow/B\right)$ amortized IOs, then queries require an
	expected $\Omega(\log_\grow N)$ IOs.
\end{theorem}

They further describe an external-memory hashing algorithm, which we refer to
here as the \defn{IP hash table}, that performs insertions in
$O\left(\frac{1}{B}\left(\lambda + \log_{\frac{M}{B}} N + \log\log
N\right)\right)$ IOs and queries in $O(\log_\lambda N)$ IOs w.h.p. Therefore,
for $\lambda=\Omega\left(\log_{M/B}N + \log\log{N}\right)$, the IP hash table
meets the tradeoff curve of \Cref{thm:ip-lower-bound} and is thus optimal.

In external memory hashing, we can assume that keys are hashed, that
is, that they are uniformly distributed and satisfy some
independence properties. The IP hash table and the following results assume
that the hash function (and therefore the hashed keys) used is $\Theta(\log
N)$-independent.

The base result of this paper is a simple external-memory
hashing scheme, the \defn{Bundle of Arrays Hash Table} (BOA),
that meets the optimal \Cref{thm:ip-lower-bound} trade off curve for
large enough $\grow$. Specifically, we show:

\begin{restatable}{theorem}{boacost}\label{thm:boa-cost}
	A \boa{} with growth factor $\grow$ supports $N$ insertions with amortized
	per entry insertion cost of $O\left(\left(\grow + \log_{\frac{M}{B}} N +
	\log_\grow{N} \right)/B\right)$ IOs and query cost of $O(\log_\grow{N})$
	IOs w.h.p.
\end{restatable}

Thus BOAs are optimal for $\grow = \Omega(\log_{\frac{M}{B}} N +
\log_\grow{N})$. 
They are readily modified to provide several variations, the
most important of which is the \defn{Bundle of Trees Hash Table}
(BOT). BOTs are
optimal for the same range of $\grow$ as the IP hash table:

\begin{restatable}{theorem}{botcost}\label{thm:hash-tree-cost}
	A BOT with growth factor $\grow$ supports $N$ insertions with amortized per
	entry insertion cost of $O\left(\left(\grow + \log_{\frac{M}{B}} N
	+ \log\log M\right)/B\right)$ IOs and a query cost of
	$O(\log_\grow{N})$ IOs w.h.p.
\end{restatable}

We further introduce the first cache-oblivious hash table, the
\defn{Cache-Oblivious Bundle of Trees Hash Table} (COBOT), which
matches the IO performance of BOT and IP hash tables.

The remainder of the paper is organized as follows. \Cref{sec:hashing}
describes the the properties of the hash function we use and \Cref{sec:lsm}
provides background on size-tiered LSMs. \Cref{sec:boa} introduces the BOAs
and \Cref{sec:bot} describes its variant, the BOT. \Cref{sec:cacheoblivious}
adapts the BOT to the cache-oblivious model, resulting in the COBOT.

\section{Preliminaries}
\subsection{Fingerprints and Hashing}\label{sec:hashing}
In order to achieve our bounds, we need $\Theta(\log N)$-wise
independent hash functions, which, once again matches IP hash tables.
We note that a $k$-wise independent hash function is also $k$-wise independent on
individual bits.  Furthermore, the following Chernoff-type bound holds:

\begin{lemma}[\cite{DBLP:journals/siamdm/SchmidtSS95}]\label[lemma]{lem:chernoff}
	Let $X_1,X_2,\ldots,X_N$ be $\lceil\mu\delta\rceil$-wise independent binary
	random variables, $X=\sum_{i=1}^N X_i$ and $\mu = \textrm{E}[X]$. Then
	\[\textrm{Pr}(X > \mu\delta) = O\left(\frac{1}{\delta^{\mu\delta}}\right),\]
	for sufficiently large $\delta$.
\end{lemma}

In what follows, if the following, we use \defn{fingerprint} to refer
to any key that has been hashed using a
$\Theta(\log N)$-wise independent hash function.Such hash functions have a compact
representation and can be specified using $\Theta(\log N)$ words. The universe
that is hashed into is assumed to have size $\Theta(N^k)$ for $k\geq2$. We ignore
collisions, but these can be handled as
in~\cite{DBLP:conf/soda/IaconoP12}.

\subsection{Log-structured Merge Trees}\label{sec:lsm}

Log-structured merge trees (LSMs) are (a family of) external-memory
dictionary data structures.  They come in two varieties:
\defn{level-tiered LSMs} (LT-LSMs) and \defn{size-tiered LSMs}
(ST-LSMs).  Both kinds are suboptimal in that they do not meet the
optimal insertion-query
tradeoff~\cite{DBLP:conf/soda/BrodalF03}, although the
COLA~\cite{DBLP:conf/spaa/BenderFFFKN07} is an optimal variant of the
LT-LSMs.  Both are popular in
practice~\cite{Ellis11,LevelDB14,RocksDB14,WiredTiger14,DBLP:conf/sigmod/SearsR12,DBLP:conf/usenix/RenG13,HBase14,DBLP:journals/tocs/ChangDGHWBCFG08,DBLP:journals/sigops/LakshmanM10,Accumulo}.

An LSM consists of sets of either B-trees or sorted arrays called \defn{runs}.
In this paper, we describe them in terms of runs, since we use runs below.

An LT-LSM consists of a cascade of levels, where each level consists
of at most one run.  Each level has a \defn{capacity} that is $\grow$
times greater than the level below it, where $\grow$ is called the
\defn{growth factor}. \footnote{Sometimes this and related structures
  are analyzed with a growth factor of $B^\epsilon$. The two are
  equivalent. We use $\grow$ rather than $\epsilon$ as the tuning
  parameter for consistency with the external-memory hashing
  literature.}  When a level reaches capacity, it is merged into the
next level (perhaps causing a merge cascade).  The amortized IO cost
for insertions is small because sequential merging is fast, although
each item will participate in $\grow/2$ merges on average.  The IO
cost for a query is high because a query must be performed indepdently
on each of $O(\log_\grow N)$ levels (although Bloom
filters~\cite{DBLP:journals/cacm/Bloom70,DBLP:journals/pvldb/BenderFJKKMMSSZ12}
are used in practice to mitigate this cost).

A ST-LSM further improves insertion IOs at the expense of queries.
Each level contains less than $\grow$ runs.  Every run on a given
level has the same size, which is $\grow$ times larger than the runs
on the level beneath it.
When $\grow$ runs are present at a level, they are merged
into one run and placed at the next level.  There are therefore
$O(\log_{\grow} N)$ levels.  Insertions are faster than in LT-LSMs
because each item is only merged once on each level.  Queries are
slower because each query must be perform $O(\grow)$ times at each
level, since the runs on each level are independent.

\section{Bundle of Arrays Hashing}\label{sec:boa}

A \defn{Bundle of Arrays Hash Table} (\boa{}) is an external hash
table based on ST-LSMs. As a first step, we show that runs with
uniformly distributed keys---for example, hashed
keys---can be searched quickly in external memory. This immediately
improves the query performance of ST-LSMs.

Interpolation search is not enough to achieve the needed improvement.
We show a balls-and-bins analysis that allows for the needed speedup.

\begin{lemma}[\cite{johnson1977urn}]\label{lem:fullestbucket}
	If $N$ balls are thrown into $Q = \Theta(N/\log{N})$ bins uniformly and
	i.i.d., then every bin has $\Theta(N/Q) = \Theta(\log{N})$ balls with high
	probability.
\end{lemma} 

\begin{lemma}\label{lem:interpolation-search}
	Let $A$ be a sorted array of $N$ uniformly distributed keys in the range
	$[0,K)$, and assume $B=\Omega(\log N)$. Then $A$ can be written to external
	memory using $O(N)$ space and $O(N/B)$ IOs so that membership in $A$ can be
	determined in $O(1)$ IOs with high probability.
\end{lemma}

\begin{proof}
	Divide the range of keys into $N/B$ uniformly sized buckets; that is,
	bucket $i$ contains keys in the range $[(i-1)KB/N,iKB/N)$. Because the keys
	in $A$ are distributed uniformly, and $B= \Omega(\log N)$, every bucket
	contains $\Theta(B)$ keys with high probability by
	\Cref{lem:fullestbucket}. Let $F$ be the number of items in the fullest
	bucket, and write the keys in each bucket to disk in order using $F$ space
	for each.  Because $F = \Theta(B)$, this takes the desired space and IOs.

	Now, to find a key, compute which bucket it belongs to. A constant number
	of IOs will fetch that bucket, whose address is known because all buckets
	have the same size.
\end{proof}

When using a ST-LSM for hashing, we can use
\Cref{lem:interpolation-search} to speed up queries:

\begin{corollary}
	If a ST-LSM contains uniformly distributed keys and has growth factor
	$\grow$, then by writing the levels as in \Cref{lem:interpolation-search},
	queries can be performed in $O(\grow \log_{\grow} N)$ IOs. The insertion
	cost is unchanged: $O\left(\frac{1}{B}\left(\log_{\grow}
	N+\log_{\frac{M}{B}} N\right)\right)$ amortized IOs.
\end{corollary} 

While the query performance has improved by a factor of $\log{N}$, the ST-LSM
is still off the optimal tradeoff curve of \Cref{thm:ip-lower-bound}. In
particular, setting $\lambda' = \log_\grow N$ and assuming that $\lambda' =
\Omega \log_{\frac{M}{B}} N$, we see that queries are at least exponentially
slower than optimal. The \boa{} will utilize additional structure in order to
reduce this query cost.

\subsection{Routing Filters}\label{sec:boa-routing-filter}
The main result of this section is an auxiliary data structure, the
\defn{routing filter} that improves the query cost by a factor of $\grow$ by
further exploiting the hashing of keys. From this point forward, we assume
that all keys in the data structure have already been hashed by a pairwise
independent hash function, and we refer to these hashed keys as
\defn{fingerprints} to make this distinction clear.  This combination of
techniques will yield a hashing data structure that is optimal for
some choices of $\grow$. In subsequent sections, we show how to
achieve optimality for a wider range of $\grow$. 

The purpose of the routing filter is to indicate probabilistically, at each
level, which run contains the fingerprint we are looking for. Each level will
have its own routing filter, defined as follows. For each level $\ell$, let
$h_\ell$ be some number, to be specified below. It will be convenient for a
fingerprint $K$ to interpret the bits of $K$ as a string of $\log_\lambda$ bit
characters, $K=K_0K_1K_2\cdots$. Let $P_\ell(K)=K_{0}\cdots K_{h_\ell}$ be the
$h_\ell$th prefix, the concatenation of the characters of $K$ up to
$K_{h_\ell}$.  The routing filter $F_\ell$ for level $\ell$ is a
$\lambda^{h_\ell}$ character array, where $F_\ell[i] = j$ if the $j$th run
$R_{\ell,j}$ contains a
fingerprint $K$ such that the $P_\ell(K)=i$, and no later run $R_{\ell,j'}$
(i.e. with $j'>j$) contains such a fingerprint. 

We also modify each run $R_{\ell, j}$ during the merge so that each
fingerprint-value pair contains a \defn{previous field} of 1 additional
character used to specify the previous run containing a fingerprint with the
same prefix or $j$ to indicate no such run exists. Thus these fingerprint-value
pairs now form a singly linked list whose fingerprint share the same prefix,
and the routing filter points to the run containing the head.

During a query for a fingerprint $K$, first $F_\ell[P_\ell(K)]$ is checked to
find the latest run containing a fingerprint with a matching prefix. Once that
fingerprint-value pair is found, its previous field indicates the next run
which needs to be checked and so on until all fingerprints with matching prefix
in the level are found.

Such routing filters induce a space/cost tradeoff. The greater $h_\ell$ is, the
more space the table takes, but the less likely it is that many runs will have
fingerprints that collide on their prefixes. The rest of this section shows
that when $h_\ell$ is set to be the base $\lambda$ log of the capacity of level
$\ell$, then that yields an optimal external hash table.

Define $\beta$, the \defn{routing table ratio}, to be the ratio of the number
of buckets in the routing filter to the size of a run. The number of entries in
a run on level $\ell$ is $B\lambda^{\ell-1}$, so $\beta =
\lambda^{h_\ell}/B\lambda^{\ell-1}$. We analyze the per-level insertion and
query cost for a given $\beta$ and $\lambda$.

\begin{lemma}\label{lem:boa-cost-level}
	For a \boa{} with growth factor $\grow$ and routing table ratio $\beta$, the
	following hold:
	\begin{enumerate}
		\item Merging a level incurs
			$\Theta\left(\frac{1}{B}\left(1 + \log_{\frac{M}{B}}\grow +
			\beta\log_{N}{\grow}\right)\right)$ IOs per fingerprint.
		\item Finding a fingerprint in a level takes $\Theta\left(1 +
			\frac{\lambda}{\beta}\right)$ IOs in expectation.
	\end{enumerate}
\end{lemma}

\begin{proof}
	\begin{enumerate}
		\item
			Merging a level requires merging together its runs as well as
			updating the next level's routing filter. Merging $\grow$ sorted
			arrays takes
			$\Theta\left(\frac{1}{B}\left(1 + \log_{M/B}\grow\right)\right)$
			IOs per fingerprint.
			
			The routing filter is updated by iterating through it and the new
			run sequentially. For each fingerprint $K$ appearing in the run,
			$F_{\ell+1}[P_{h_{\ell+1}}(K)]$ is copied to the previous field in the
			run, and $F_{\ell+1}[P_{h_{\ell+1}}(K)]$ is set to the number of the
			current run. Each entry in the routing filter is a character, the
			routing filter has $\beta$ characters for each new fingerprint, so
			it requires $\Theta\left(\frac{\beta}{B}\log_{N}{\grow}\right)$
			IOs per fingerprint to update sequentially. 

		\item
			To query for a fingerprint $K$, first the routing filter is
			checked, which takes $O(1)$ IOs, and then the runs with fingerprints
			matching the prefix of $K$ are checked.

			Given some enumeration of the fingerprints in level $\ell$, denote
			the $i$th fingerprint by $K_i$. Let $X_i$ be the indicator random
			variable which is 1 if $P_{h_\ell}(K)=P_{h_\ell}(K_i)$ and 0
			otherwise. Because the hash function is pairwise independent, $K$
			and $K_i$ are uniformly distributed and their bits are pairwise
			independent. Thus $\E{X_i} \leq \frac{1}{\lambda^{h_\ell}}$. The
			expected number of fingerprints in the level with prefix
			$P_{h_\ell}(K)$ is at most $\sum_{i=1}^{B \lambda^\ell}\E{X_i} \leq
			\frac{B \lambda^\ell}{\lambda^{h_\ell}} = \frac{\grow}{\beta}.$ By
			\Cref{lem:interpolation-search}, each of these fingerprints can be
			found and checked in $O(1)$ IOs. Thus, the expected per-level query
			cost is $O\left(1 + \frac{\grow}{\beta}\right)$.
	
	\end{enumerate}
\end{proof}

\begin{lemma}\label{lem:boa-cost-beta}
	A \boa{} with growth factor $\grow$ and routing table ratio $\beta$ has
	insertion cost $O\left(\frac{1}{B}\left(\beta + \log_{\frac{M}{B}} N +
	\log_\lambda N\right)\right)$ and query cost $O\left(\left(1 +
	\frac{\grow}{\beta}\right)\log_\grow{N}\right)$.
\end{lemma}

\begin{proof}
	Because a \boa{} has $\log_\lambda N$ levels, this follows immediately from
	\Cref{lem:boa-cost-level}.
\end{proof}

So for a fixed $\grow$, there is no advantage to choosing $\beta =
\omega(\grow)$. On the other hand, $\beta = o(\grow)$ is suboptimal, because
then choosing $\beta'=\grow'=\beta$ changes a linear factor in the query cost
to a logarithmic one. Therefore, it is optimal to choose $\beta =
\Theta(\grow)$, and in what follows we will fix $\beta = \grow$. Now
the main theorem follows immediately:

\boacost*

Thus, a BOA is optimal for large enough $\lambda$:

\begin{corollary}
	Let $B$ be a \boa{} with growth factor $\grow$ containing $N$ entries. If
	$\grow = \Omega\left(\log_\frac{M}{B}{N} +
	\frac{\log{N}}{\log\log{N}}\right),$ then $B$ is an optimal unsorted
	dictionary.
\end{corollary}

Note that the condition that $\grow=\Omega(\log_{M/B}{N})$ is related to the
permutation bound~\cite{DBLP:journals/cacm/AggarwalV88}. This is because BOAs
and their variations support some form of successor operation for the order of
the fingerprints (the hashed order of the keys). 

\section{Bundle of Trees Hashing}\label{sec:bot}

In order for a BOA to be an optimal dictionary, its growth factor $\grow{}$
must be $\Omega(\log{N}/\log\log{N})$. Otherwise, the cost of insertion is
dominated by the cost of merging, which in slow because it effectively sorts
the fingerprints using a $\lambda$-ary merge sort. In this section, we present
the \defn{Bundle of Trees Hash Table} (BOT), which is a BOA-like structure. A
BOT stores the fingerprints in a log in the order in which they arrive. Each
level of the BOT is like a level of a BOA, where the bundle of arrays on each
level is replaced by an search structure on the log (the \defn{routing tree})
and a data structure needed to merge routing trees (the \defn{character
queue}). The character queue performs a delayed sort on the characters needed
at each level, thus increasing the arity of the sort and decreasing the IOs. 

A BOT has $s = \lceil \log_\lambda N/B \rceil$ levels, each of which consists
of at most one routing tree where the root has degree less than $\lambda$ and
all internal nodes have degree $\lambda$. Each node of a routing tree contains
a routing filter, which functions similarly to the routing filters in
\Cref{sec:boa}. In a BOT, the routing filter takes as input a fingerprint and
outputs a set of pointers to the children which may contain it, though some of
these may be false positives. It also returns some auxiliary information
discussed below, which together with the child pointer is referred to as the
\defn{sketch} of the fingerprint.

Each leaf points to a block of $B\log_\lambda N$ fingerprints in the log; the
reason for using blocks of this size will be explained in
\Cref{sec:character-queue}. The deepest level $s$ indexes the beginning of
the fingerprint log with a tree of depth $s$, the next level then indexes the
next section and so forth.  Insertions are
appended to the log until they form a block, at which point they are added to
the tree in the 1st level of the BOT.

\IInsist{Querying a BOT}
A query to the BOT for a fingerprint $K$ is performed indepedently at each
level, beginning at the root of each routing tree.  At a routing node of height
$h$ in a tree, the routing filter is queried. Whereas the routing filter in a
BOA returns only the last array containing a fingerprint with a given prefix
$P_h(K)$, BOTs use a \defn{refined routing filter} that returns a full list of
all the children with such a fingerprint. The details of the refined routing
filter are left for \Cref{sec:bot-query}. The query is then passed to each of
these children, until it reaches a block of the log, which is then searched in
full.  In this way queries are ``routed'' down the tree on each level to the
part of log where the fingerprint and its associated value are. In addition, as
queries descend the routing tree, they may generate false positives which are
likewise routed down towards the log.

An issue that arises from this querying algorithm is that if a query generates a
false positive in a node of height $h$, that is, there is a fingerprint $K'$
with $P_h(K)=P_h(K')$, then in the child containing $K'$, the shorter prefixes
will also match, {\it i.e.} $P_{h-1}(K)=P_{h-1}(K')$. Therefore false positives
will always propagate down the routing tree and at each subsequent node may in
turn generate more false positives. To prevent this, the routing filter of each
node of the routing tree keeps, for each fingerprint $K$, an additional
\defn{check character} taken from the tail of K. Positive queries must also
match this character, and nodes of different heights use different parts of the
fingerprint for the check characters so that the probabilities of two
fingerprints matching on different levels are independent. Check characters are
explained further in \Cref{sec:bot-query}.

\IInsist{Inserting into a BOT}
When a level $i$ in the BOT fills, its routing tree is merged into the routing
tree of level $i+1$, thus increasing the degree of the target routing tree by 1 (and
perhaps filling level $i+1$, which triggers a merge of level $i+1$ into
$i+2$, and so on). The merge of level $i$ into level $i+1$ consists of adding
the prefix-sketch pairs of the fingerprints from level $i$ to the routing
filter of the root on level $i+1$. The child pointers of these pairs will point
to the root of the formerly level-$i$ routing tree, so it becomes a child of
the root of the level $i+1$ routing tree, although it isn't moved or copied.
In this way, a BOT resembles an LT-LSM,
described in \Cref{sec:lsm}.

In order to add a fingerprint $K$ from level $i$ to the root routing filter on
level $i+1$, the prefix $P_{i+1}(K)$ must be known. However, the root routing
filter on level $i$ only stores the prefix $P_i(K)$ for each fingerprint $K$ it
contains, so that in particular the last character of $P_{i+1}(K)$ is missing.
As described in \Cref{sec:character-queue}, each level has a character queue,
which provides this character, as well as the check characters, in order to
merge the routing trees efficiently.

By replacing the arrays of a BOA on each level by character queues, the BOT can
insert efficiently for a larger range of $\grow{}$. Because the arrays are no
longer available to answer queries, BOTs instead use the recursively
constructed routing trees, which require some refinement over the routing
filters in BOAs.  With these in place, however, query performance is optimal,
and the BOT becomes an optimal dictionary for a wider range of the parameter
$\grow{}$, matching the range of the IP hash
table~\cite{DBLP:conf/soda/IaconoP12}.

\subsection{BOT Queries}\label{sec:bot-query}

This section covers the routing tree structure in more detail. First, we cover
the specifics of check characters, and then we introduce the refined routing
filter and prove its performance characteristics.

As described above, in a BOT each false positive in a node of a routing tree
queries an additional child.  Because the routing filter in the child has
shorter prefixes, that false positive will cause an entire path down the tree
to be accessed. Moreover, along the way more false positives can be generated.
Check characters are used to reduce the probability of false positives
in BOTs and short-circuit the paths that do occur.

The $i$th check character $C_i$ of a fingerprint $K$ is the $i$th character
from the end of the string representation of $K$. As described in
\Cref{sec:hashing}, we assume that the fingerprints are taken from a universe
of size at least $N^2$ so that the check characters do not overlap with the
characters used in the prefixes of the routing filters.

The routing filter of a node of height $h$ in a routing tree stores $C_h(K)$ in
the sketch for each fingerprint $K$ in the filter, and when queried, returns a
list of the sketches of prefix-matching fingerprints in the order that they
appear in the log.

Now, the query only proceeds on those children whose check characters match the
check character $C_h(K)$. Since the characters of the fingerprint are uniformly
distributed, the check character of each false positive matches with
probability $1/\grow$. Moreover, the characters of each level are
non-overlapping, so for fingerprints $K$, $K'$ the event that $V_h(K)=V_h(K')$
is independent of the event that $V_{h-1}(K)=V_{h-1}(K')$.  Therefore a false
positive on a node of height $h$ may still be eliminated in its child (of
height $i-1$), short-circuiting the paths that false positives would otherwise
create.

While not strictly necessary, in order to simplify the analysis, we will
further arrange it so that false positives may only be created in the root of
the tree, and that at each level, only a $1/\lambda$ fraction of the false
positives survive in expectation.  To prevent new false positives from being
generated when a query passes from a parent to a child, we also keep the
\defn{next character} of each fingerprint in its prefix-sketch pair stored in
the routing filter. For a fingerprint $K$ in a node of height $h$, the next
character is just the next character that follows the prefix, $P_h(K)$, so that
its prefix in the parent, $P_{h+1}(K)$, can be obtained. A false positive in
the children which is not in the parent will not match this next character and
can be eliminated.

When there are multiple prefix-matching fingerprints in both a parent and its
child, we would like to be able to align the lists returned by the routing
filters so that known false positives in the parent (either from check or next
characters) can be eliminated in the child. Otherwise the check character in
the child of a known false positive in the parent may match the queried
fingerprint, and therefore more than $1/\lambda$ of the false positives may
survive in expectation.  To this end, we require the routing filter to return
the list of sketches of prefix-matching fingerprints in the order they appear
in the log.  Then after the sketches in the child list whose next characters do
not match the parent are elimated, the remaining phrases will be in the same
order as in the parent. In this way, known false positives can also be
eliminated in the child.

Now because of the next characters, false positives may only be created in the
root of the routing tree. Each false positive in the root corresponds to a
fingerprint $K'$ in the level. At each node on the path to $K'$'s location in
the log, we use the ordering to determine which returned sketch corresponds to
$K'$, so that the false positive corresponding to $K'$ is eliminated with
probability $1/\lambda$. Thus the query path for $K'$ causes at most
$\sum_{i=1}^h 1/\lambda^i=O\left(\frac{1}{\lambda}\right)$ node accesses. This
is independant of the number of false positives in the root. Since there are
$O(1)$ expected false positives in the root, we have shown:

\begin{lemma}\label[lemma]{lem:routing-tree-false-positive}
	During a query to a routing tree, the expected number of nodes accessed due
	to false positives is at most $O\left(\frac{1}{\lambda}\right)$.
\end{lemma}

\IInsist{Refined Routing Filter.}
A BOA routing filter handles prefix collisions by returning only the last run
containing the queried fingerprint and then chaining in the runs. Because there
are no longer arrays with which to chain, the BOT routing filter, however, must
handle prefix collisions itself and return a complete ordered list of sketches
for all prefix-matching fingerprints, while having the same performance as in
\Cref{sec:boa-routing-filter}. 

The idea behind the refined routing filter is to keep the prefix-sketch pairs
in a list, and use a hash table on prefixes to point queries to the appropriate
place. Each pointer may require as many as $\Omega(\log N)$ bits, and we
require the routing filter to have $O(1)$ characters per fingerprint, so we
require the hash table to use shorter prefixes so as to reduce the number of
buckets and thus reduce its footprint. In particular, it uses prefixes which are
$\log_\lambda\log_\lambda N$ characters shorter, which we refer to as \defn{pivot
prefixes}.

The list \defn{delta encodes} the prefix $P_h(K)$ for each fingerprint $K$, together with the
sketch, $S_h(K)$. This means the difference between $P_h(K)$ and the preceding
prefix is stored, together with the sketch. In addition, the first entry
following each pivot prefix contains the full prefix $P_h(K)$, rather then just
the difference. Otherwise, when the hash table routes a query to that place in
the list, the full prefix wouldn't be computable.

We first analyze the space efficiency of a delta encoded list of a collection of
prefixes and then analyze the performance characteristics of refined routing
filters.

\begin{lemma}\label[lemma]{lem:delta-encoding}
	A list of delta-encoded prefixes with density $D$, that is there are $D$
	prefixes in the list for every possible prefix, requires
	$O(-\log_\lambda{D})$ characters per prefix.
\end{lemma}

\begin{proof}
	The average difference between consecutive prefixes is $1/D$. Because
	logarithms are convex, the average number of characters required to
	represent this difference is therefore $O(-\log_\lambda{D})$.
\end{proof}

Now we can prove:

\begin{lemma}\label[lemma]{lem:bot-routing-filter-update}
	A refined routing filter can be updated using
	$O\left(\frac{\grow\log\grow{}}{B\log N}\right)$ IOs per new entry.
\end{lemma}

\begin{proof}
	Let $C$ be the capacity of the level. There are $\frac{C}{\log_\lambda N}$
	pivot prefixes. For each pivot prefix, the hash table stores the bit
	position in a list with at most $C$ entries, where $C \leq N$.  Each entry
	is at most $\log N$ bits, so this position can be written using $O(\log N)$
	bits.

	For each fingerprint in the node, the list contains $O(1)$ characters by
	\Cref{lem:delta-encoding}, or $O(\log \lambda)$ bits. Additionally, each
	pivot prefix has to an initial entry of length $O(\log N)$ bits, so the
	list all together uses $O(C\log\lambda + \frac{C}{\log_\lambda N} \cdot
	\log N)= O(C\log\grow{})$ bits.

	When the refined routing filter is updated, the old version is read
	sequentially and the new version is written out sequentially. $C/\lambda$
	fingerprints are added at a time, so this incurs
	$O\left(\frac{\grow\log\grow{}}{B}\right)$ IOs per entry.
\end{proof}

The refined routing filter also still performs constant IO lookups:

\begin{lemma}\label[lemma]{lem:bot-routing-filter-lookup}
	A refined routing filter performs lookups in $O(1)$ IOs in expectation.
\end{lemma}

\begin{proof}
	The pivot bit string of a fingerprint and its successor are accessed from
	the hash table in $O(1)$ IOs. This return beginning and ending bit
	positions in the list. Because the fingerprints are distributed uniformly
	and are pairwise independent there
	are $O(\log_\lambda N)$ fingerprints matching the pivot prefix in
	expectation.  The list has $O(\log \lambda)$ bits per fingerprint, so
	$O(1)$ words are fetched from the list in expectation, and hence $O(1)$
	IOs.
\end{proof}

\subsection{Character Queue}\label{sec:character-queue}

The purpose of the character queue is to store all the sketches of fingerprints
contained in a level $i$ that will be needed during a merge in the future.
When level $i$ is merged into level $i+1$, the character queue outputs a sorted
list of the delta-encoded prefix-sketch pairs of all the fingerprints, which is
used to update the root routing tree. The character queue is then merged into
the character queue on level $i+1$.

The character queue effectively performs a merge sort on the
sketches. If it were to merge all the sketches as soon as they are
available, this would consist of $\lambda$-ary merges.  In order to
increase the arity of the merges, it defers merging sketches which are
not needed immediately. The sketches are collection of \defn{series},
by which we mean a collection of sorted runs. Each series stores a
continuous range of sketches $S_i(K),S_{i+1}(K),\ldots,S_{i+j}(K)$ for
each fingerprint $K$, together with the prefix up to the first sketch,
$P_{i-1}(K)$. These prefixes are delta encoded in their run. Thus the
size of an entry is determined by the number of sketches in the range
and the length of the prefix relative to the size of the run (by
\Cref{lem:delta-encoding}).

\IInsist{The character queue tradeoff}
We are faced with the following tradeoff. If the character queue merges a series
frequently, the delta encoding is more efficent, which decreases the cost of
the merging. However the arity is lower, which increases it. The character
queue uses a merging schedule which balances this tradeoff and thus
achieves optimal insertions.

\IInsist{The character queue merging schedule} The character
queue on level $i$ contains the sketches
$S_{i+1}(K), S_{i+2}(K),\ldots S_s(K)$ of each fingerprint $K$ in the
level. These characters are stored in a collection of series
$\{\sigma_{j_q}\}$, where $j_q$ is the smallest multiple of $2^q$
greater than $i$. Series $\sigma_{j_q}$ contains the sketches
$S_{j_q}(K),\ldots,S_{j_{q+1}-1}(K)$. Each series consists of a
collection of sorted runs each of which stores the delta encoded
prefix of each fingerprint together with its sketches.

When level $i$ fills, the runs in the series $\sigma_{i+1}$ are merged, and the
character queue outputs the delta encoded prefix-sketch pairs,
$(P_{i+1}(K),S_{i+1}(K))$ to update the root routing filter on level $i+1$. If
$2^{\rho(i+q)}$ is the greatest power of 2 dividing $i+1$ ($\rho$ is sometimes
referred to as the \defn{ruler function}~\cite{wiki:Thomae's_function}), then
$\sigma_{i+1}$ also contains the next $2^{\rho(i+1)}-1$ sketches of each
fingerprint. These are batched and delta encoded to become runs in the series
$\sigma_{j_q}$ for $q = [0,\rho(i+1)]$.

This leads to the following merging pattern: $\sigma_j$ batches $2^{\rho(j)}$
sketches, and has delta encoded prefixes of $2^\rho(j)$ characters on average,
by \Cref{lem:delta-encoding}.  Therefore,

\begin{lemma}\label[lemma]{lem:character-queue-size}
	A series $\sigma_j$ in a character queue contains $O(2^{\rho(j)})$
	characters per fingerprint.
\end{lemma}

This leads to a merging schedule where the characters per item merged on the
$j$th level is $O(2^{\rho(j)})$. Starting from 1 this is
$1,2,1,4,1,2,1,8,1,2,1,4,1,2,1,16,\ldots$, which resemble the tick marks of a
ruler, hence the name ruler function.

We now analyze the cost of maintaining the character queues.

\begin{lemma}\label[lemma]{lem:character-queue-update}
	The total per-insertion cost to update the character queues in a BOT is
	$\Theta\left(\frac{1}{B}\left(\log_{\frac{M}{B}} N + \log\log
	M\right)\right)$.
\end{lemma}

\begin{proof}
	When $\sigma_j$ is merged, $\lambda^{2^\rho(j)}$ runs are merged, which has
	a cost of
	$O\left(\frac{2^{\rho(j)}}{B}\left\lceil\log_{M/B}\left(\lambda^{2^\rho(j)}\right)\right\rceil\right)$
	characters per fingerprint.

	There are $\log_\lambda \frac{N}{B} = O(\log_\lambda N)$ levels, so this leads to the
	following total cost in terms of characters:

	\begin{align*}
		O\left(\sum_{i=1}^{\log_\lambda N} 2^{\rho(j)}\left\lceil\log_{\frac{M}{B}}\left(\lambda^{2^\rho(j)}\right)\right\rceil\right)
		&= O\left(\sum_{k=0}^{\log \log_\lambda N} \frac{\log_\lambda N}{2^k} \cdot 2^k \left\lceil\log_{\frac{M}{B}}\left(\lambda^{2^k}\right)\right\rceil\right) \\
		&= O\left(\log_\lambda N \left(\log\log M + \sum_{k=\log\log M}^{\log \log_\lambda N} 2^k\log_{\frac{M}{B}}\lambda\right)\right) \\
		&= O\left(\log_\lambda N\left(\log\log M + \log_{\frac{M}{B}}N\right)\right),
	\end{align*}
	where the last equality is because the RHS sum is dominated by its last
	term. Because there are $\log_\lambda N$ characters in a word, and all
	reads and writes are performed sequentially in runs of size at lease $B$,
	the result follows.
\end{proof}

The character queue is where we require that the blocks of the log have size
$B\log_\lambda N$, because we want the runs created when the block is added to
the first level to be at least size $B$. 

\subsection{Performance of the BOT}

We can now prove Theorem 3:

\botcost*

\begin{proof}
	By \Cref{lem:bot-routing-filter-update}, the cost of updating the routing
	filters is $O\left(\frac{\grow}{B}\right)$, since there are $O(\log_\grow
	N)$ levels. This together with the cost of updating the character queues,
	given by \Cref{lem:character-queue-update}, is the insertion cost.

	By \Cref{lem:routing-tree-false-positive}, a query for fingerprint $K$
	accesses an average of $O\left(\frac{\log_\lambda N}{\lambda}\right)$ nodes
	across the routing trees on all level due to false positives. Using
	\Cref{lem:chernoff} with $\delta=\lambda$, we have that this is
	$O(\log_\lambda N)$ nodes w.h.p.
		
	If $K$ is contained in the bot, then $O(\log_\lambda N)$ nodes are accessed
	on its root-to-leaf path.

	A block of the log is scanned at most once for a true positive and also
	whenever a false positive from the level $i$ root survives $i$ times. The
	expected number of such false positives for level $i$ is $1/\grow^i$, so
	the expected number across levels is $O\left(\frac{1}{\lambda}\right)$
	Therefore by \Cref{lem:chernoff}, the probability that $\omega(1)$ blocks
	are scanned due to false positives is $O\left(\frac{1}{\lambda}\right)$.
	Each block can be scanned in $O(\log_\lambda N)$ IOs, so this yields the
	result.
\end{proof}

It follows that:

\begin{corollary}\label[corollary]{cor:bot}
	Let $\mathcal{B}$ be a BOT with growth factor $\lambda$
	containing $N$ entries. If $\lambda = \Omega\left(\log_{\frac{M}{B}} N +
	\log\log M\right)$, then $\mathcal{B}$ is an optimal dictionary.
\end{corollary}

\section{Cache-Oblivious BOTs}\label{sec:cacheoblivious}

In this section, we show how to modify a BOT to be cache oblivious.  We call
the resulting structure a cache-oblivious hash tree (COBOT). 

Much of the structure of the BOT translates directly into the cache-oblivious
model. However, some changes are necessary. In particular, when the series of
character queues are merged, this merge must be performed cache-obliviously
using funnels~\cite{DBLP:conf/focs/FrigoLPR99}, rather than with an (up to)
$M/B$-way merge. Also, the log cannot be buffered into sections of size
$O(B\log_\lambda N)$, and so instead they are buffered into sections of
constant size, items are immediately added to routing filter, and the extra IOs
are eliminated by optimal caching.

When an insertion is made into a CO hash tree, its fingerprint-value pair is
appended to the log, and it is immediately inserted into level 1. Thus, the
leaves of the routing trees point to single entries in the log.

The series of the character queues must be placed more carefully as well. In
particular, for each $j$, the runs of series $\sigma_j$ must be laid out
back-to-back, so that even when they are short, they may be read efficiently
across the level.

The series are merged using a \defn{partial funnelsort}. Funnelsort is a
cache-oblivious sorting algorithm that makes use of
$K$-funnels~\cite{DBLP:conf/focs/FrigoLPR99}. A $K$-funnel is a CO data structure
that merges $K$ sorted lists of total length $N$. We make use of the the
following lemma.

\begin{lemma}[\cite{DBLP:conf/focs/FrigoLPR99}]\label[lemma]{lem:funnel}
	A $K$-funnel merges $K$ sorted lists of total length $N\geq K^3$ in
	$O\left(\frac{N}{B}\log_{M/B} \frac{N}{B} + K + \frac{N}{B}\log_{K}
	\frac{N}{B}\right)$ IOs, provided the tall cache assumption that
	$M=\Omega(B^2)$ holds.
\end{lemma}
	
The partial funnelsort used to merge $K$ runs of a series with total length $L$
(in words) performs a single merge with a $K$-funnel if $L \geq K^3$ and
recursively merges the run in groups of $K^{1/3}$ runs otherwise.

\begin{corollary}\label[corollary]{cor:funnelsort}
	A partial funnelsort merges $K$ runs of total word length $L$ in
	$O\left(\frac{L}{B}\log_{M/B} \frac{L}{B} + \frac{L}{B}\log_{K}
	\frac{L}{B}\right)$ IOs, provided the tall cache assumption that
	$M=\Omega(B^2)$ holds.
\end{corollary}

\begin{proof}
	The base case of the recursion occurs either when there is only 1 list
	remaining or the remaining lists fit in memory. In any other case of the
	recursion, since $L=\Omega(B^2)$ by the tall cache assumption, the $K$ term
	in \Cref{lem:funnel} is dominated.

	The recurrence is dominated by the cost of the funnel merges, which yields
	the result.
\end{proof}

\begin{theorem}
	If $M=\Omega(B^2)$, then a CO hash tree with $N$ entries and growth factor
	$\lambda$ has amortized insertion cost
	$\Theta\left(\frac{1}{B}\left(\lambda + \log\log M + \log_{M/B}
	N/B\right)\right)$ and query cost $\Theta\left(\log_\lambda N\right)$,
	w.h.p.
\end{theorem}

\begin{proof}
	We may assume that the caching algorithm sets aside enough memory that the
	last $B$ items in the log, together with the subtree rooted at their least
	common ancester, are cached. Thus the log is updated at a per-item cost of
	$O(1/B)$.

	The proof of \Cref{thm:hash-tree-cost} now carries over to the CO hash
	tree. The routing filters are updated the same way, and the cost of
	updating the character queues is unchanged, by \Cref{cor:funnelsort}. 

	Queries are performed as in \Cref{sec:bot-query}, except that now the level
	1 nodes cover $O(1)$ fingerprints, but the depth of the tree is unchanged,
	so the cost is the same.
\end{proof}

\end{document}